%% file: paper.tex
\documentclass[twoside,leqno]{article}

\usepackage[letterpaper,left=1in,right=1in,top=1in,bottom=1in]{geometry}
\usepackage{hyperref}
\usepackage{amsthm}
\usepackage{amsmath}
\usepackage{amsfonts}
\usepackage{tikz}
\usepackage{subcaption}
\usepackage{accents}
\usepackage{wrapfig}
\usepackage{todonotes}
\usepackage{textpos}

\pdfoutput=1

\newtheorem{theorem}{Theorem}
\newtheorem{lemma}[theorem]{Lemma}
\newtheorem{Conjecture}[theorem]{Conjecture}
\newtheorem{Definition}[theorem]{Definition}

\newcommand*{\email}[1]{\href{mailto:#1}{\nolinkurl{#1}}} 

\newcommand{\bO}{\ensuremath{O}}
\newcommand{\btO}{\ensuremath{O^*}}

\newcommand{\Q}{\ensuremath{\mathbb{Q}}}

\newcommand{\dcup}{\,\dot{\cup}\,}

\newcommand{\Ra}{\mathbf{R}}
\newcommand{\AR}{\underaccent{\sim}{\mathbf{R}}}

\newtheorem{claim}[theorem]{Claim}
\newenvironment{claimproof}{\begin{proof}}{\end{proof}}

\begin{document}


\newcommand\relatedversion{}

\title{\Large Fast Deterministic Chromatic Number under the\\{} Asymptotic Rank Conjecture\relatedversion}
\author{%
Andreas Bj\"orklund\thanks{ IT University of Copenhagen, Denmark. \email{anbjo@itu.dk}.
Supported by the VILLUM Foundation, Grant~16582.} 
\and 
Radu Curticapean\thanks{\rightskip=5cm University of Regensburg, Germany and IT University of Copenhagen, Denmark. \email{radu.curticapean@ur.de}.
Funded by the European Union (ERC, CountHom, 101077083). Views and opinions expressed are however those of the author(s) only and do not necessarily reflect those of the European Union or the European Research Council Executive Agency. Neither the European Union nor the granting authority can be held responsible for them.}%
\and 
Thore Husfeldt\thanks{Basic Algorithms Research Copenhagen and IT University of Copenhagen, Denmark. \email{thore@itu.dk}.
Supported by the VILLUM Foundation, Grant~16582.}
\and
Petteri Kaski\thanks{Aalto University, Finland. \email{petteri.kaski@aalto.fi}}
\and 
Kevin Pratt\thanks{New York University, USA. \email{kpratt@andrew.cmu.edu}}
}

\date{}
	
	\maketitle
	\begin{textblock}{5}(8.5, 8.35) \includegraphics[width=120px]{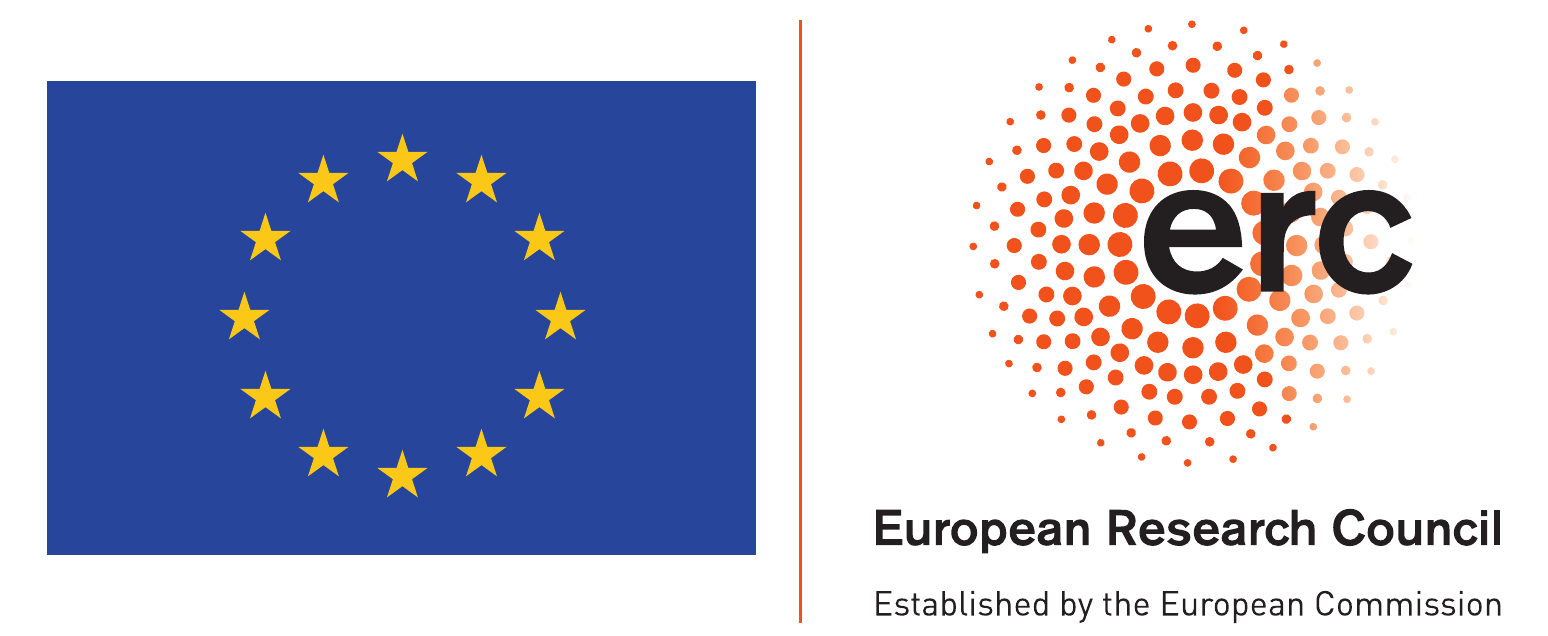} \end{textblock}

	\begin{abstract}\noindent
	In this paper we further explore the
	recently discovered connection by Bj\"orklund and Kaski [STOC~2024] and Pratt [STOC 2024] between the asymptotic rank conjecture of Strassen [Progr.~Math.~1994] and the three-way partitioning problem. We show that under the asymptotic rank conjecture, the chromatic number of an $n$-vertex graph can be computed deterministically in $\bO(1.99982^n)$ time, thus giving a conditional answer to a question of Zamir [ICALP 2021], and questioning the optimality of the $2^n\operatorname{poly}(n)$ time algorithm for chromatic number by Bj\"orklund, Husfeldt, and Koivisto [SICOMP 2009].

	Viewed in the other direction, if chromatic number indeed requires deterministic algorithms to run in close to $2^n$ time, we obtain a sequence of explicit tensors of superlinear rank, falsifying the asymptotic rank conjecture. 
	
	Our technique is a combination of earlier algorithms for detecting $k$-colorings for small $k$ and enumerating $k$-colorable subgraphs, with an extension and derandomisation of Pratt's tensor-based algorithm for balanced three-way partitioning to the unbalanced case. 
	\end{abstract}
	
	\section{Introduction}
		
	Recently, Bj\"orklund and Kaski~\cite{BjorklundK24} discovered a contention between two seemingly unrelated conjectures:
	Strassen's \emph{asymptotic rank conjecture} from the area of algebraic complexity~\cite{Strassen1991,Strassen1994}, and the \emph{set cover conjecture} from the area of fine-grained complexity. The former claims a strong upper bound on the complexity of tensors under the Kronecker product operation, and can be understood as a generalization of the conjecture that the exponent of matrix multiplication $\omega$ equals $2$. 
	The latter, on the other hand, claims a strong lower bound on the complexity of the set cover problem. Specifically, it postulates that given any $\varepsilon > 0$, there exists a $k$ such that given as input an integer $t$ and a set family $\mathcal F$ of $k$-sized subsets of $[n]$, one needs time $(2-\varepsilon)^n$ to decide whether $[n]$ can be covered with $t$ sets from $\mathcal F$. Bj\"orklund and Kaski~\cite{BjorklundK24} showed that these conjectures are inconsistent: if the asymptotic rank conjecture is true, there exists an $\varepsilon > 0$ such that for any constant set size $k$, set cover can be solved in time $(2-\varepsilon)^n$. Pratt \cite{Pratt2024} subsequently showed that this holds for any $\varepsilon < 2-3/2^{2/3} \approx 0.11$.
	
	The results of \cite{BjorklundK24,Pratt2024} are consequences of hypothetical fast algorithms for the \emph{almost-balanced three-way partitioning problem}, where we are given a universe $U$ of $n$ elements and families $\mathcal F_1,\mathcal F_2$, and $\mathcal F_3$ of subsets each of size at most $\nu n$ for $\nu$ close to $\tfrac13$ (in the case of \cite{Pratt2024}, equal to $\tfrac13$), and would like to decide if there are three pairwise disjoint $f_1\in \mathcal F_1,f_2\in \mathcal F_2,f_3\in \mathcal F_3$ such that $f_1\cup f_2 \cup f_3=U$. Throughout, we say that a set family $\mathcal F$ over a universe $U$ is $\nu$-\emph{bounded} if no set in $\mathcal F$ is larger than $\nu |U|$.
	
	Motivated by an application to graph coloring, we extend the results of \cite{BjorklundK24,Pratt2024} to the case when $\nu$ is permitted to be arbitrarily close to $1/2$. Moreover, our hypothetical algorithm is \emph{deterministic}, in contrast to the randomized algorithms of \cite{BjorklundK24,Pratt2024}. Our main algorithmic component is the following:
	\begin{theorem}[Main; Deterministic fine-grained three-way partitioning]
		\label{thm:fine-3part}
		If the asymptotic rank conjecture is true over any field of characteristic zero, 
		then for every constant $\frac{1}{3}\leq\nu<\frac{1}{2}$, and 
		every constant $\epsilon>0$, there is a deterministic algorithm that solves 
		the three-way partitioning problem over an $n$-element universe
		in $\bO\bigl(\binom{n}{\lfloor\nu n\rfloor}^{1+\epsilon}\bigr)$ time for
		$\nu$-bounded set families.
	\end{theorem}
	Since the input families may have size $\binom{n}{\lfloor \nu n \rfloor}$, Theorem~\ref{thm:fine-3part} is close to best possible.
	Some consequences of previous work are immediate:
	With $\nu = \frac{1}{3}$, Theorem~\ref{thm:fine-3part} yields a deterministic version of the three-way partitioning algorithm of~\cite{Pratt2024}. It then follows from the constructions in \cite{BjorklundK24,Pratt2024} that if the asymptotic rank conjecture is true, there are deterministic algorithms with running time $\bO(1.89^n)$ for the following two problems:
	\begin{enumerate}
	  \item set cover for a universe of $n$~elements, provided that the subsets have constant size,
	  \item detecting a Hamiltonian cycle in an $n$-vertex directed graph.
	\end{enumerate}
	The derandomisation is particularly interesting for the second problem, 
	because finding a deterministic $(2-\epsilon)^n$-time algorithm for Hamiltonicity is an outstanding open problem, even in undirected graphs.

	\medskip
	We turn our attention to set cover instances where the subsets can have more than a constant number of elements.
	As anticipated by a remark in \cite{BjorklundK24}, 
	we combine Theorem~\ref{thm:fine-3part} with ideas that are implicit in previous work by Bj\"orklund, Husfeldt, Kaski, and Koivisto~\cite{BjorklundHKK10}  on listing $k$-coverable subsets.
	The result is a deterministic set cover algorithm where the subset size can be arbitrarily close to~$\frac14n$:
	
	\begin{theorem}[Set cover with large subsets under the asymptotic rank conjecture]
		\label{thm: setcover}
		Let $\delta < \frac14$.
		If the asymptotic rank conjecture is true over a field of characteristic zero, then any set cover instance $([n],\mathcal F,t)$ where  $\mathcal F$ is $\delta$-bounded, can be solved deterministically in $\btO((2-\epsilon)^n)$ time, for some $\epsilon>0$.
	\end{theorem}
	The corresponding earlier result \cite[Theorem 4]{BjorklundK24}, establishes this (with a randomized algorithm) for much smaller subsets;
	the admittedly coarse estimate in \cite[p.~14]{BjorklundK24} has $\delta=10^{-7}$.

	\medskip
	However, our main (and guiding) application is graph coloring.
	Unlike Theorems~\ref{thm:fine-3part} and \ref{thm: setcover},
	our coloring result makes no assumptions about subset size.

	To be precise, the $k$-coloring problem is to find a proper coloring of the vertices of an undirected graph on $n$~vertices using at most $k$~colors.
	The smallest $k$ for which such a $k$-coloring exists is the graph's \emph{chromatic number}.
	Graph coloring is a set cover problem, where the universe consists of the graph's vertices, and the implicitly provided set family are the independent sets.
	It is therefore immediate from Theorem~\ref{thm: setcover} that there is a $(2-\epsilon)^n$-time algorithm for chromatic number when each independent set contains sufficiently less than $\frac14n$~vertices.
	However, we show that the size constraint can be avoided.
	The result is an algorithm for the general chromatic number problem:

	\begin{theorem}[Chromatic number under the asymptotic rank conjecture]
		\label{thm: chromaticnumber}
		If the asymptotic rank conjecture is true over a field of characteristic zero, then the chromatic number of an $n$-vertex graph can be computed deterministically in $\bO(1.99982^n)$ time. 
	\end{theorem}

	The best known general algorithm for $k$-coloring and chromatic number runs in $\btO(2^n)$; see Bj\"orklund, Husfeldt, and Koivisto~\cite{BjorklundHK2009}. 
	Zamir~\cite{Zamir21} posed as an open problem the question of whether there are $\bO((2-\epsilon_k)^n)$-time algorithms for $k$-coloring for every $k$, where $\epsilon_k>0$ is a constant depending only on $k$. He proved that this is indeed the case for $k\leq 6$, extending earlier results for $k\leq 4$.
	In a subsequent paper, Zamir~\cite{Zamir23} showed that such algorithms for deciding $k$-colorability exist for any fixed~$k$, provided that the graph is \emph{regular} or very close to regular. Yet, we do not know whether $7$-coloring in a general graph, nor if chromatic number even in a regular graph, can be computed in $\bO((2-\epsilon)^n)$ time. 
	Theorem~\ref{thm: chromaticnumber} establishes that such algorithms \emph{do} exist under the asymptotic rank conjecture.
	
	It is crucial for our construction behind Theorem~\ref{thm: chromaticnumber} that we have the freedom to choose $\nu$ close to $\tfrac12$ in Theorem~\ref{thm:fine-3part}. 
	This helps us detect a $k$-coloring that is balanced enough to admit a partitioning of the color classes into three parts each sufficiently smaller than $\frac12 n$ in total size.
	To cope with $k$-colorings that are not balanced in the above sense,
	we also need some $k$-coloring algorithms from the literature to complete the construction.
	To be concrete, we use the deterministic $4$-coloring algorithm by Fomin, Gaspers, and Saurabh~\cite{Fomin2007}.
	However, any deterministic algorithm for $4$-coloring running in time $\bO((2-\epsilon_4)^n)$ for some $\epsilon_4>0$ would suffice to establish a version of Theorem~\ref{thm: chromaticnumber} with a $(2-\epsilon)^n$-time bound.
	
	\subsection{Strassen's Asymptotic Rank Conjecture}
	\label{sec: intro-ARC}
	Our Theorem \ref{thm:fine-3part} is based on a hypothetical fast algorithm for evaluating certain trilinear forms, which we show exists under Strassen's asymptotic rank conjecture. We now give a brief introduction to this conjecture and its algorithmic importance, starting with standard preliminaries and conventions on tensors.
	
	Throughout this paper we work over the field $\Q$ of rationals. A trilinear form, or for short, a \emph{tensor}, is a multilinear map $T : \Q^n \times \Q^n \times \Q^n \to \Q$. We view $T$ concretely as a degree 3 polynomial in 3 disjoint sets of variables, where each monomial is a product of one variable from each variable set:
	\[T = \sum_{i,j,k \in [n]} a_{ijk} X_iY_jZ_k,\]
	for coefficients $a_{ijk} \in \Q$.
	
	We say that a tensor $T$ has rank one if $T = ( \sum a_i X_i )(\sum b_i Y_i)(\sum c_iZ_i)$ for some $a,b,c \in \Q^n$. The \emph{tensor rank} of $T$, denoted $\Ra(T)$, is the minimum number $r$ such that $T$ can be expressed as an $\Q$-linear combination of $r$ rank-one tensors. The \emph{Kronecker product} of $T = \sum_{i,j,k \in [n]} a_{ijk}X_iY_jZ_k $ and $T' = \sum_{i,j,k \in [n']} b_{ijk}X_iY_jZ_k$ is the trilinear form $T \otimes T' : \Q^{nn'} \times \Q^{nn'} \times \Q^{nn'} \to \Q$ given by
	\[T \otimes T' := \sum_{\substack{i,j,k\in [n]\\i',j',k'\in[n']}} a_{ijk}b_{i'j'k'}X_{i,i'}Y_{j,j'}Z_{k,k'}.\]
	We use the shorthand $T^{\otimes r}$ to denote the Kronecker product of $r$ copies of $T$.
	
	Tensor rank is submultiplicative under Kronecker product: $\Ra(T \otimes T') \le \Ra(T) \Ra(T')$. As a result, Fekete's subadditivity lemma implies that the \emph{asymptotic rank}~\cite{Gartenberg1985} of $T$,
	\[\AR(T) = \lim_{r \to \infty} \Ra(T^{\otimes r})^{1/r},\]
	is well-defined; see e.g.~\cite{WigdersonZ2023}. The algorithmic importance of this quantity is that the number of arithmetic operations needed to evaluate $T^{\otimes r}$ equals $\AR(T)^{r(1+ o(1))}$. This is folklore, but for completeness we provide a proof in Lemma \ref{lem:trilinear-arc}.
	
	Asymptotic rank arises naturally in the study of algorithms for fast matrix multiplication. Let $M_n = \sum_{i,j,k \in [n]} X_{ij}Y_{jk}Z_{ki}$ be the $n \times n$ \emph{matrix multiplication tensor}. Strassen showed that the exponent of matrix multiplication $\omega$ satisfies $\omega = \lim_{n \to \infty} \log_{2^n} \Ra(M_{2^n})$~\cite{Strassen1986,Strassen1988}. Moreover, these tensors obey the simple but fundamental identity $M_n \otimes M_m = M_{nm}$; in particular, $M_{2^r} = M_2^{\otimes r}$. Combining these two facts it follows that $\log_2 \AR(M_2) = \omega$. Understanding $\omega$ is thus equivalent to understanding the asymptotic rank of the constant-size tensor $M_2$.
	
	\begin{Conjecture}[Asymptotic rank conjecture~\cite{Strassen1994,BurgisserCS2013}
		\footnote{Strassen's original conjecture~\cite[Conjecture~5.3]{Strassen1994} assumes that the tensor $T$ is \emph{tight} in the sense that for some choice of 
			bases of $\Q^n$ there exist injective functions 
			$f, g, h : [n] \to \mathbb{Z}$ such that $f(i) + g(j) + h(k) = 0$ 
			for all monomials $X_iY_jZ_k$ in the support of $T$. More recent versions (cf.~B\"urgisser, Clausen, and Shokrollahi~\cite[Problem~15.5]{BurgisserCS2013}; 
			also e.g.~Conner, Gesmundo, Landsberg, Ventura, and 
			Wang~\cite[Conjecture~1.4]{ConnerGLVW2021} 
			as well as Wigderson and Zuiddam~\cite[Section 13, p.~122]{WigdersonZ2023}) 
			of the conjecture omit the tightness assumption. We state the conjecture without the tightness assumption but with the understanding that the tensors to which we apply the conjecture in this work are tight.}{}]
		For every tensor $T : \Q^n \times \Q^n \times \Q^n \to \Q$, we have $\AR(T) \le n$.
	\end{Conjecture}
	
	This generalizes the conjecture that $\omega = 2$, as it would imply that $\AR(M_2) \le 4$, so $\omega = 2$. We mention a couple of pieces of evidence in favor of this conjecture:
	\begin{enumerate}
		\item For any tensor $T$, Strassen showed that $\AR(T) \le n^{2\omega/3} < n^{1.6}$~\cite{Strassen1988}
		\footnote{The result is implicit in the proof of Proposition 3.6 in \cite{Strassen1988}; cf.~also \cite[Proposition~2.12]{ChristandlVZ2021} and \cite[Remark~2.1]{ConnerGLV2022}.}{}, whereas ``typical" tensors have rank $\Theta(n^2)$; e.g.~\cite{Landsberg2012}. Thus Kronecker-structured tensors do have significantly lower rank than most.
		
		\item The analogue of the asymptotic rank conjecture for bilinear forms is the statement that multiplication of a Kronecker-structured matrix with a vector can be done in nearly linear time. This is true by e.g.~Yates's algorithm~\cite{Yates1937}.
	\end{enumerate}
	
	\begin{figure}
                \begin{center}
		\includegraphics[width=12cm]{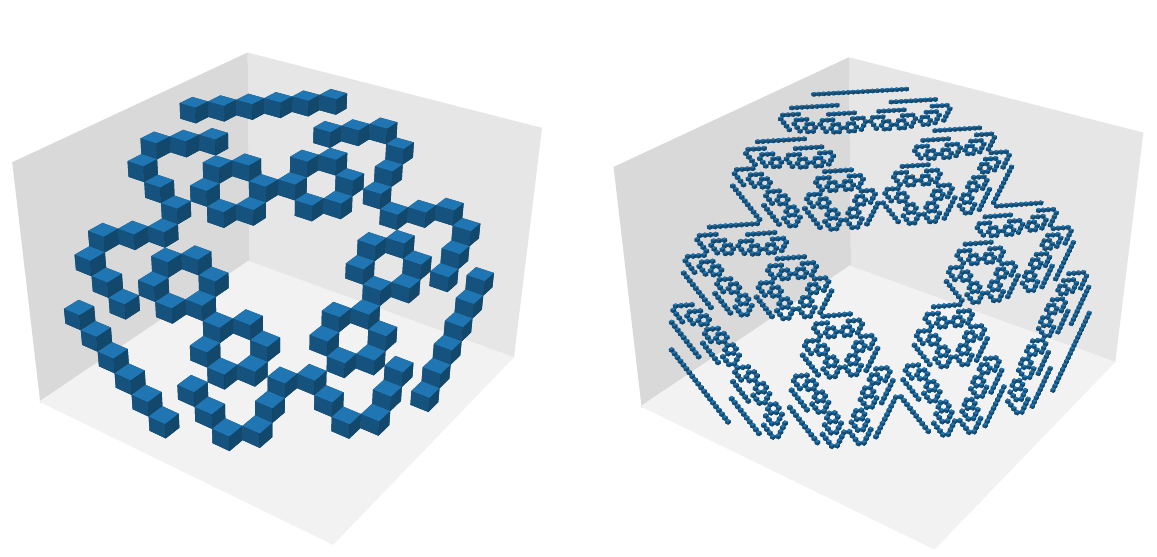}
                \end{center}
		\caption{The partitioning tensors $T_{1/3,\,6}$ and $T_{1/3,\,9}$.}
		\label{fig: partition}
	\end{figure}
	Of interest to us will be (a slight variant of) the $\tau$-\emph{bounded partitioning tensor}
	\[
	T_{\tau,k}=\sum_{\substack{I,J,K\subseteq [k]\\|I|,|J|,|K|\leq \lceil\tau k\rceil\\I\dcup J\dcup K=[k]}}X_IY_JZ_K.\,
	\]
	See Figure~\ref{fig: partition} for an illustration. Notice that the evaluation of $T_{\tau,k}$ at $X,Y$, and $Z$ being the indicator vector of a set family $\mathcal{F}$ equals the number of $\tau$-bounded tripartitions of $\mathcal{F}$. So if it were the case that $T_{\tau,k}^{\otimes r} = T_{\tau,rk}$, analogous to what happened in the case of the matrix multiplication tensors, we would be able to evaluate $T_{\tau,k}$ with $(\sum_{i=0}^{\tau k} \binom{k}{i})^{1+\varepsilon}$ operations, and Theorem~\ref{thm:fine-3part} would immediately follow. Unfortunately this is not the case, as the support of $T_{\tau,k}^{\otimes r}$ is a strict subset of that of $T_{\tau, rk}$. Nevertheless, we show how to repair this deficiency using the ``breaking and randomizing" idea of Karppa and Kaski~\cite{karppa2019probabilistic}, which we additionally modify in a way that makes the ``randomizing'' part amenable to derandomization via pairwise independence. Ultimately we end up with a hypothetical fast algorithm for evaluating a trilinear form with nonnegative coefficients and the same support as $T_{\tau,k}$, which suffices for the problem of detecting $\tau$-bounded tripartitions.

        \subsection{Earlier Work on Chromatic Number and $k$-Coloring}
	\label{sec: related work}
	Exponential time algorithms for chromatic number and $k$-coloring have been studied for more than fifty years.
	The first theoretical non-trivial algorithm for chromatic number that we know of is an $\btO(n!)$ time algorithm by Christofides~\cite{Christofides1971} from 1971.
	Using exponential space dynamic programming across the vertex subsets, Lawler~\cite{Lawler1976} designed an algorithm running in $\bO(2.4423^n)$ time based on enumerating the maximal independent sets of the graph. Further refinements of this idea by bounding the number of maximal independent sets by size led to the improvements $\bO(2.4151^n)$ by Eppstein~\cite{Eppstein2003} and $\bO(2.4023^n)$ by Byskov~\cite{Byskov2003}. Bj\"orklund and Husfeldt~\cite{BjorklundH2006} used the principle of inclusion and exclusion combined with fast matrix multiplication to compute the chromatic number in $\bO(2.3236^n)$ time. By replacing the fast matrix multiplication for the fast zeta transform, a variant of Yates's algorithm, the computation was brought down to $\btO(2^n)$ time by Bj\"orklund, Husfeldt, and Koivisto~\cite{BjorklundHK2009}.

	Faster algorithms exist for $k$-coloring for small $k$. The number of maximal independent sets in a graph is at most $3^{n/3}$, as proved by Moon and Moser~\cite{MoonM1965}, and the maximal independent sets can be listed this fast~\cite{TomitaTT2006}. Checking for each such set whether the complementary vertex set induces a bipartite graph gives a deterministic $\bO(1.4423^n)$ time algorithm for $3$-coloring. 
	Schiermeyer~\cite{Schiermeyer93} gave a $\bO(1.415^n)$ time algorithm. 
	Beigel and Eppstein~\cite{BeigelE2005} constructed a $\bO(1.3289^n)$ time algorithm for $3$-coloring based on a new algorithm for $(3,2)$-CSP. In a recent preprint, Meijer~\cite{Meijer2023} claims a $\bO(1.3217^n)$ time algorithm.
        
        For $4$-coloring, Byskov, Madsen, and Skjernaa~\cite{ByskovMS2005} presented an $O(1.8613^n)$ time algorithm that is based on listing all maximal bipartite subgraphs and, for each, checking if the vertex complement induces a bipartite graph. 
	Beigel and Eppstein~\cite{BeigelE2005} presented a general $\bO(1.8072^n)$ time algorithm for $(4,2)$-CSP, which in particular can be used to solve $4$-coloring. 
	Byskov~\cite{Byskov2004} gave a faster algorithm in $\bO(1.7504^n)$ time based on listing maximal independent sets and reducing to $3$-coloring. Fomin, Gaspers, and Saurabh~\cite{Fomin2007} gave a $\bO(1.7272^n)$ time algorithm that also lists maximal independent sets, but additionally exploits the fact that either branching is sufficiently efficient, or the remaining graph has low path width.  The running time bound of this algorithm was recently improved to $\bO(1.7215^n)$ by Clinch, Gaspers, Saffidine, and Zhang~\cite{Clinch2024}. A deterministic polynomial space algorithm for $4$-coloring running in $\bO(1.7159^n)$ time was very recently obtained by Wu, Gu, Jiang, Shao, and Xu~\cite{WuGJSX2024}.

	\subsection{Organization of the Paper}
The rest of this paper is organized as follows. We postpone our proofs of Theorem~\ref{thm:fine-3part} and Theorem~\ref{thm: setcover} and explain the connection to chromatic number first. In Section~\ref{sec: preliminaries}, we set up some notation and preliminaries including how we can use Theorem~\ref{thm:fine-3part} to solve for balanced coverings. In Section~\ref{sec: chromatic number}, we describe and analyse five cases leading up to the proof of Theorem~\ref{thm: chromaticnumber}. 
In Section~\ref{sec: set cover under arc}, we finally prove Theorem~\ref{thm:fine-3part} and Theorem~\ref{thm: setcover}.
	\section{Preliminaries}
	\label{sec: preliminaries}
	\subsection{Notation} We write $\btO(f(n))$ for a non-decreasing function $f(n)$ of the input size parameter $n$ to suppress polynomial factors in $n$. With $A\dcup B$ we denote the disjoint union of two sets $A$ and $B$. 
	For numbers $\delta >0$ and $q$, we abbreviate $(1\pm \delta)q = [(1-\delta)q,(1+\delta)q]$.
	
	\subsection{Entropy and Binomial Coefficients} 
        Let us recall the binary entropy function
	\[
	H(p)=-p\log_2p-(1-p)\log_2 p\,,\qquad 0<p<1\,,
	\]
	which is monotone increasing on $0<p<\frac{1}{2}$.
        For integers $1\leq b\leq a-1$, we recall (e.g.,~\cite{Jukna2011}) that  
	\begin{equation}\label{eq:binomentropy}
		\binom{a}{b}=\frac{a!}{b!(a-b)!}\leq 2^{H(b/a)a}\,.
	\end{equation}

	\subsection{Yates's Algorithm} 
        We recall the following standard application of Yates's~\cite{Yates1937} algorithm for evaluating a Kronecker-structured trilinear form. 
	
        \begin{lemma}
        \label{lem:trilinear-arc}
        Let $c$ be a positive integer constant and let $T=\sum_{i,j,k\in[c]}T_{i,j,k}X_iY_jZ_k$ be a tensor with rational coefficients $T_{i,j,k}$ for $i,j,k\in[c]$. Then, assuming the asymptotic rank conjecture over a field of characteristic zero, for all $\epsilon>0$ there exists an algorithm that computes the polynomial evaluation $T^{\otimes r}(\alpha,\beta,\gamma)$ in time $\bO(c^{(1+\epsilon)r})$ from three vectors $\alpha,\beta,\gamma:[c]^r\rightarrow\{0,1\}$ given as input.
        \end{lemma}

        \begin{proof}
        By standard properties of tensor rank, there exist 
        rational matrices $A,B,C\in\Q^{[c^2]\times [c]}$ of shape $c^2\times c$ 
        with $\sum_{\ell\in[c^2]}A_{\ell,i}B_{\ell,j}C_{\ell,k}=T_{i,j,k}$ 
        for all $i,j,k\in[c]$. Let $\epsilon>0$ be fixed. By the asymptotic
        rank conjecture, there exist positive integer constants $s$ and $d$
        as well as rational\footnote{A priori, we only know that the entries of these matrices are elements of our working field $\mathbb{F}$, which may be an extension of $\Q$. However it follows from \cite[Proposition 15.17, Corollary 15.18]{BurgisserCS2013} that we may assume that they are in fact elements of the subfield $\Q$.} matrices $\tilde A,\tilde B,\tilde C\in\Q^{[d]\times [c]^s}$ of shape $d\times c^s$ with $c^s<d\leq c^{(1+\epsilon/2)s}$ and 
        $\sum_{\ell\in[d]}A_{\ell,I}B_{\ell,J}C_{\ell,K}=\prod_{u\in[s]}T_{i_u,j_u,k_u}=(T^{\otimes s})_{I,J,K}$ for all $I=(i_1,i_2,\ldots,i_s)\in[c]^s$, $J=(j_1,j_2,\ldots,j_s)\in[c]^s$, and $K=(k_1,k_2,\ldots,k_s)\in[c]^s$. 
        Given three vectors $\alpha,\beta,\gamma:[c]^r\rightarrow\{0,1\}$ as input, the 
        algorithm computes, using Yates's algorithm as a subroutine for 
        multiplying the Kronecker-structured matrix with the vector 
        in each case, the three vectors
        \begin{equation}
        \label{eq:yates-3}
        \begin{split}
        \hat\alpha&\leftarrow\bigl(\tilde A^{\otimes\lfloor r/s\rfloor}\otimes A^{\otimes(r-s\lfloor r/s\rfloor)}\bigr)\alpha\,,\\
        \hat\beta&\leftarrow\bigl(\tilde B^{\otimes\lfloor r/s\rfloor}\otimes B^{\otimes(r-s\lfloor r/s\rfloor)}\bigr)\beta\,,\\
        \hat\gamma&\leftarrow\bigl(\tilde C^{\otimes\lfloor r/s\rfloor}\otimes C^{\otimes(r-s\lfloor r/s\rfloor)}\bigr)\gamma\,.\\
        \end{split}
        \end{equation}
        Finally, the algorithm computes the return value 
        \begin{equation}
        \label{eq:trilinear-hat}
T^{\otimes r}(\alpha,\beta,\gamma)\leftarrow \sum_{\lambda\in [d]^{\lfloor r/s\rfloor}\times[c^2]^{r-s\lfloor r/s\rfloor}}\hat\alpha_\lambda\hat\beta_\lambda\hat\gamma_\lambda\,.
        \end{equation}
        Because $c$, $d$, and $s$ are constants, the steps \eqref{eq:yates-3} and \eqref{eq:trilinear-hat} can be implemented with at most $\bO(c^{(1+\epsilon/2)r})$ rational operations. We observe that the entire computation can thus be implemented in $\bO(c^{(1+\epsilon)r})$ time; indeed, since each of the input vectors $\alpha,\beta,\gamma$ is zero-one-valued, and the rational matrices $A,B,C,\tilde A,\tilde B,\tilde C$ each have constant size and constant-bit denominator complexity, we observe that the denominator-complexity in all the rational arithmetic in \eqref{eq:yates-3} and \eqref{eq:trilinear-hat} is at most $\bO(r)$ bits. 
        \end{proof}

	\subsection{Listing $t$-Covers}
	Given a set family $\mathcal{F}\subseteq2^{[n]}$ and $t \in \mathbb N$,
	a tuple $(F_{1},\ldots,F_{t})\in\mathcal{F}^{t}$
	is a \emph{$t$-cover} of a set $U$ if $U\subseteq\bigcup_{i=1}^{t}F_{i}$.
	We say that $U$ is \emph{$t$-covered} by $\mathcal{F}$ if it admits a $t$-cover.
	
	Let $\mathcal{F}\subseteq\mathcal{U}\subseteq2^{[n]}$ be set families.
	Bj\"orklund \emph{et al.}~\cite{BjorklundHKK10} showed how to efficiently list the sets in $\mathcal U$ that are $t$-covered by $\mathcal{F}$, provided that $\mathcal U$ is closed under subsets. For completeness, we include a self-contained proof of this result. 
	
	\begin{lemma}\label{lem: list-covers}
	Let $\mathcal{F}\subseteq\mathcal{U}\subseteq2^{[n]}$ with $\mathcal{U}$ closed under subsets and computable in time $O^{*}(|\mathcal{U}|)$.
	For $t\in \mathbb N$, 
	we can list all $U \in \mathcal{U}$ that are $t$-covered by $\mathcal F$
	in time $O^{*}(|\mathcal{U}|)$.
	
	\end{lemma}

        \begin{proof}
        The algorithm is based on efficient algorithms for certain transforms over $\mathcal U$, which is seen as part of the subset lattice.
	The relevant transforms are essentially ``upward'' and ``downward'' sums for vectors that are indexed by $\mathcal U$.
	\begin{claim}
	\label{lem: transforms}
		Given $\alpha\in\mathbb{Q}$ and a vector $f$ with indices from $\mathcal U$,
		there is an $O^{*}(|\mathcal{U}|)$ time algorithm for computing the
		vectors $T_{\alpha}^{\downarrow}f$ and $T_{\alpha}^{\uparrow}f$
		defined by
		\[
		(T_{\alpha}^{\downarrow}f)(X)=\sum_{Y\subseteq X,\,Y\in\mathcal{U}}\alpha^{|X\setminus Y|}f(Y),
		\qquad
		(T_{\alpha}^{\uparrow}f)(X)=\sum_{Y\supseteq X,\,Y\in\mathcal{U}}\alpha^{|Y\setminus X|}f(Y).
		\]	
	\end{claim}
        \begin{claimproof}
		We express the entries of $T_{\alpha}^{\downarrow}f$ and $T_{\alpha}^{\uparrow}f$ as path counts in directed acyclic graphs $D$ and use dynamic programming to determine all relevant path counts simultaneously in the claimed running time.
		Given a directed acyclic graph $D$ with edge-weights $w:E(D)\to\mathbb{Q}$
		and given vertices $u,v\in V(D)$, define the weighted path count $\rho(u,v):=\sum_{P}\prod_{e\in P}w(e)$, where $P$
		ranges over all paths from $u$ to $v$.
		
		We first consider $T_{\alpha}^{\uparrow}$; the algorithm for $T_{\alpha}^{\downarrow}$ is similar.
		Construct the directed acyclic graph $D$ with vertex set $\{\bot\}\cup[n+1]\times\mathcal{U}$,
		where $(i,X)$ for $X\in\mathcal{U}$ is connected to $(i+1,X)$ with
		weight $1$ and to $(i+1,X\setminus\{i\})$ with weight $\alpha$; 
		this vertex exists by the assumption on $\mathcal U$.
		For each $Y\in\mathcal{U}$, assign weight $f(Y)$ to the edge from $\bot$ to $(1,Y)$.
		We observe that the vertices of $D$ can be arranged into $n+2$ layers, with $\bot$ forming the first layer.
		For all $X \in \mathcal U$, we have $\rho(\bot,(n+1,X))=(T_{\alpha}^{\uparrow}f)(X)$.
		Proceeding forward along layers, we inductively determine $\rho(\bot,v)$
		for $v\in V(D)$ from the values of the at most two predecessors of $v$. This
		takes at most two operations per $v\in V(D)$, summing to $\btO(|\mathcal{U}|)$
		operations overall. 
		
		To obtain $T_{\alpha}^{\downarrow}$, construct a ``converse'' of $D$ by adding a vertex $\top$ instead of $\bot$, with edges of weight $f(Y)$ from $(n+1,Y)$ to $\top$. Then determine $\rho((1,X),\top)$ as above.
	\end{claimproof}

	Now, let $\mathcal{F}_{\downarrow}=\bigcup_{F\in\mathcal{F}}2^{F}$ be
	the downward closure of $\mathcal{F}$. We compute its indicator vector $f$ by viewing $\mathcal{F}$
	as an indicator vector and clipping non-zeroes in $T_{1}^{\uparrow}\mathcal{F}$ to $1$. 
	
	Then compute $g=T_1^{\downarrow}f$. 
	For $X \in \mathcal U$, we have that $g(X)=\#\{I\subseteq X\mid I\in\mathcal{F}_{\downarrow}\}$
	counts the subsets of $X$ that are $1$-covered by $\mathcal F$.
	Let $g^t$ be the pointwise $t$-th power of $g$. Then $g^{t}(X)=g(X)^{t}$
	for $X\in\mathcal{U}$ counts the $t$-tuples $I=(I_{1},\ldots,I_{t})\in\mathcal{F}_{\downarrow}^{t}$
	with $I_{i}\subseteq X$. We call any such tuple $I$ a \emph{$t$-pre-cover}
	of $X$ and observe that $I$ is a $t$-cover of $X$ if $I$
	is \emph{not} a $t$-pre-cover of $X\setminus\{x\}$ for any $x\in X$.
	By the inclusion-exclusion formula, the number of $t$-covers
	of $X$ is thus $h^{t}(X):=\sum_{Y\subseteq X}(-1)^{|X\setminus Y|}g^{t}(Y)$.
	We can compute the entire vector $h^t$ as $T_{-1}^{\downarrow}g^{t}$.
	The algorithm then outputs all $X$ for which the number of $t$-covers satisfies $h^{t}(X)>0$.
	Using the above proposition to compute the transforms efficiently, the running time follows.
\end{proof}

	\subsection{Detecting $\nu$-Balanced $k$-Covers under the Asymptotic Rank Conjecture}
	We proceed to show that the listing algorithm from the previous section, 
	together with the conjectured tripartition algorithm from Theorem~\ref{thm:fine-3part}, 
	yields a nontrivial $k$-covering algorithm that works for general $k$, provided that the cover satisfies a specific balancing condition.

	Let $\frac13\leq\nu<\frac12$.
	A $k$-cover $(F_1,\ldots,F_k)$  of $[n]$ is \emph{$\nu$-balanced} if it can be partitioned into three roughly equal-sized parts in the following sense:
	There are positive integers $k_1,k_2,k_3$ with $k_1+ k_2+k_3= k$ such that 
	\[ 
	|F_1\cup\ldots\cup F_{k_1}|\leq 
	|F_{k_1+1}\cup\ldots\cup F_{k_1+k_2}|\leq
	|F_{k_1+k_2+1}\cup\ldots\cup F_k|
	\leq \nu n\,.\]

	We now have the following algorithm for $\nu$-balanced $k$-covers.

	\begin{lemma}
	  \label{lem:balancedkcover}
	  Let $\frac13 \leq\nu < \frac12$ and let $\epsilon > 0$.
	  For positive integer $k\in\mathbb N$ and set family $\mathcal F \subseteq 2^{[n]}$,
	  if the asymptotic rank conjecture is true
	  then there is a deterministic algorithm to decide if $[n]$ admits a $\nu$-balanced $k$-cover from $\mathcal F$ with running time $\btO\bigl(\binom{n}{\lfloor \nu n\rfloor}^{1+\epsilon}\bigr)$.
	
	\end{lemma}

        \begin{proof}
	 
	  Set $s=\lfloor \nu n \rfloor$.
	  Begin by constructing the $\nu$-bounded family $\mathcal U = \{\,U\subseteq [n]\colon |U|\leq s\,\}$
	  of size $\binom{n}{0}+\cdots+\binom{n}{s} \leq n\binom{n}{s}$.
	  Similarly, let $\mathcal F'$ be the  $\nu$-bounded subfamily of $\mathcal F$,
	  formally $\mathcal F' = \{\, F\in \mathcal F\colon |F| \leq s\,\}$.
	  Clearly, any $k$-cover by $\mathcal F'$ is also a $k$-cover by $\mathcal F$.
	  For the other direction, assume that $(F_1,\ldots, F_k)$ is a $\nu$-balanced $k$-cover.
	  Then in particular, every set size $|F_i|$ is at most $s$, so that each $F_i$ belongs to $\mathcal F$.
	  Thus, it suffices to investigate the covers from $\mathcal F'$.

	  For each $t\in\{1,\ldots, k\}$,
	  let $\mathcal U_t$ denote the subsets $U\in\mathcal U$ that are $t$-covered by $\mathcal F'$.
	  By Lemma~\ref{lem: list-covers} these can be constructed in time $\btO(\binom{n}{s})$.
	  Now iterate over all positive integers $k_1,k_2,k_3$ with $k_1+k_2+k_3=k$; there are $\binom{k+2}{2}=O(n^2)$ such triples.
	  For each such $k_1,k_2,k_3$ 	we can invoke Theorem~\ref{thm:fine-3part} on the three-way partitioning instance $([n], \mathcal U_{k_1}, \mathcal U_{k_2}, \mathcal U_{k_3})$ to find a tripartition $f_1\in U_{k_1},f_2\in U_{k_2},f_3\in U_{k_3}$ of $[n]$.
	  By construction, each part $f_i$ admits a $k_i$-cover, which together cover $[n]$.
	  The running time is $\btO\bigl( \binom{n}{\lfloor \nu n\rfloor}^{1+\epsilon}\bigr)$.
	\end{proof}

	\section{Chromatic Number under the Asymptotic Rank Conjecture}
	\label{sec: chromatic number}
	A (proper) $k$-\emph{coloring} of an undirected $n$-vertex graph $G$ is a mapping $f\colon V(G)\rightarrow \{1,\ldots,k\}$
	such that $f(u)\neq f(v)$ for $uv\in E(G)$.
	The vertex subsets $S_i$ for which $f(S_i) = i$ are called \emph{color classes};
	each color class forms an independent set in $G$.

	\subsection{Balanced Colorings and Size Profiles}

	We begin by pointing out two immediate consequences of the results in the previous section.
	
	\begin{lemma}
		\label{lem: listinginduced}
		Given a graph $G$ on $n$ vertices, a positive integer $k$, and $0<\nu<\tfrac12$, we can list every vertex subset $X\subset V(G)$ with $|X|\leq \nu n$ such that the induced graph $G[X]$ is $k$-colorable, in time $\btO \left(\binom{n}{\lfloor \nu n \rfloor}\right)$.
        \end{lemma}
        
	\begin{proof}
	  Identify $V(G)$ with $[n]$, set $\mathcal U=\binom{[n]}{\leq \nu n}$, set $t=k$, and apply Lemma~\ref{lem: list-covers} with $\mathcal F$ chosen as the family of independent sets in $G$.
	\end{proof}

	\begin{lemma}
		\label{lem: 3waykcol}
		Let $\frac13\leq \nu<\frac12$ and let  $\epsilon>0$.
		If the asymptotic rank conjecture is true, we can given any large enough graph $G$ on $n$ vertices detect whether it admits a $k$-coloring such that the color classes can be partitioned in three parts, none of which is larger than $\nu n$, in $\btO(\binom{n}{\lfloor \nu n\rfloor}^{1+\epsilon})$ time. 
	\end{lemma}
	
	\begin{proof}
	  Identify $V(G)$ with $[n]$ and apply Lemma~\ref{lem:balancedkcover} with $\mathcal F$ as the family of independent sets in $G$.
	\end{proof}

	There are many colorings that this algorithm alone will fail to detect.
	One example is the complete $4$-partite graph $K_{n/4,n/4,n/4, n/4}$, which admits only $\frac14$-balanced colorings.
	Other examples are very `lopsided' colorings, where one of the color classes is huge.

	To classify these cases, we need some terminology.
	Our convention will be to list the independent sets in nonascending order by size, breaking ties arbitrarily, so that $|S_1|\geq \cdots \geq |S_k|$.
	The \emph{size profile} of a coloring $f$ is the sequence $s_1,\ldots, s_k$ of cardinalities $s_i=|S_i|$.
	
	An independent set is \emph{maximal under set inclusion} if it is not a proper subset of another independent set.
	Such a set is called a \emph{maximal independent set}.
	We recall the following:
	\begin{lemma}
		\label{lem:S1mis}
		If $G$ is $k$-colorable then it admits a $k$-coloring with color classes $S_1,\ldots, S_k$, where $S_1$ is a maximal independent set.
	\end{lemma}
	\begin{proof}
		Consider a coloring $f$ whose largest color class~$S_1$ is a subset of the maximal independent set $X$. 
		Then $f'$ defined by $f'(x)= 1$ for $x\in X$ and $f'(v) =f(v)$ for $v\in V\setminus X$ is a $k$-coloring with largest color class $X$.
	\end{proof}

	Therefore, to find a $k$-coloring it suffices to investigate all possible size profiles under the assumption that $S_1$ is a maximal independent set.

	Depending on the shape of a coloring's size profile, we will apply different algorithms to detect it.
	To classify the shapes into useful categories we observe the following:
	
	\begin{lemma}
		\label{lem:cases}
		Consider $k$ nonnegative numbers $s_1,\ldots, s_k$ with
		\( s_1+\cdots +s_k = n\).
		Then for any number $d$, at least one of the following is true:
		\begin{enumerate}
			\renewcommand{\theenumi}{\alph{enumi}}
			\item $k\leq 3$.
			\item $s_1+s_2+s_3+s_4 >  n- 6d$.
			\item $s_1+s_2 >  \frac12 n + d$.
			\item $\frac12n - d <  s_1 \leq  \frac12 n+ d$ and $s_2 < 2d$.
			\item $s_1 \leq \frac12n-d$ and $s_1+s_2+s_3+s_4 \leq n-6d$.
		\end{enumerate}
	\end{lemma}
	
	\begin{proof}
		Assume (a)--(d) fail.
		Note that when (c) fails we have $s_1\leq s_1+s_2\leq \frac12n +d$, so for (d) to fail we must have $s_1\leq \frac12 n -d$
		or $s_2\geq 2d$.
		In the latter case, we have
		\[ s_1 = s_1 + s_2 - s_2 \leq \tfrac12n + d -2d = \tfrac12n -d,\]
		as well.
		In either case (e) applies.
	\end{proof}
	
	Each of the five cases will be handled separately by an algorithm in the next five subsections. According to the above lemma, these algorithms together exhaust all size profiles and therefore all possibilities for a $k$-coloring.
	
	To gain some intuition, set $d = \frac1{145} n$. We use the $4$-coloring algorithm in Fomin~\emph{et al.}~\cite{Fomin2007} and refer to it as \textbf{$\mathbf 4$-Colorability}. 
	
	\subsection{Case A: At Most Three Color Classes}
	\label{sec:chi3}
	For a given graph and $k\leq 3$, there are algorithms that work for every size profile and every $d$.
	For $k=1$, we inspect the graph to see whether it has no edges, in which case $\chi(G)=1$.
	For $k=2$ we check bipartiteness by depth-first search to determine if $\chi(G)=2$.
	For $k=3$, we can run any $3$-coloring algorithm.
	Thus if we let $t_3 (n)$ denote the running time of the best $3$-coloring algorithm, we can bound the running time of this case to within a polynomial factor of 
	\begin{equation}\label{eq:ta}
		t_{\mathrm a}(n) = t_3(n)\,.
	\end{equation}
	
	\subsection{Case B: Four Very Large Color Classes}
	\label{sec: unbalanced}
	We detect $k$-colorings in which the four largest color classes exhaust almost all vertices.
	To be precise, we want to detect $k$-colorings whose size profile $(s_1,\ldots, s_k)$ satisfies
	\begin{equation}
		\label{eq: 4sum}
		s_1+s_2+s_3+s_4 > n-6d\,.
	\end{equation}

	\subsubsection*{Algorithm}
	Given $G$, $k$, and $d$, we invoke Lemma~\ref{lem: listinginduced} to list every vertex subset $X$ of size at most $6d$ such that $G[X]$ is  $(k-4)$-colorable.
	For each such $X$ we investigate if the graph $G[V\setminus X]$ induced by the complementary vertex set is $4$-colorable using the algorithm \textbf{$\mathbf 4$-Colorability}.

	\subsubsection*{Analysis}
	Writing $t_4(r)$ for the running time of the best algorithm for $4$-colorability on $r$-vertex graphs, we get
	\begin{equation}\label{eq:rt4}
		t_{\mathrm b}(n, d) = \btO\left(\binom{n}{\lfloor 6d\rfloor}\right)+
		\sum_{i=1}^{\lfloor 6d\rfloor} \binom{n}{i} t_4(n-i) .
	\end{equation}
	
	\subsection{Case C: Two Large Color Classes}
	
	We turn to the case where the two largest color classes exhaust at least a little more than half of the vertices.
	To be precise, we want to detect $k$-colorings whose size profile $(s_1,\ldots, s_k)$ satisfies
	\begin{equation}
		\label{eq: 2sum}
		s_1+s_2 > \tfrac12n+ d.
	\end{equation}
	
	\subsubsection*{Algorithm}
	Given $G$, $k$, and $d$,
	we invoke Lemma~\ref{lem: listinginduced} for all vertex subsets of size at most $\frac{n}{2}-d$ to see which subsets $X$ induce a graph $G[X]$ of chromatic number at most $k-2$, and then check if $G[V\setminus X]$ is bipartite using depth-first search.
	
	\subsubsection*{Analysis}
	The running time is within a polynomial factor of
	\begin{equation}\label{eq:rt2}
		t_{\mathrm c}(n, d) = \sum_{i=1}^{\lfloor \frac12n-d\rfloor}\binom{n}{i} \leq n\binom{n}{\lfloor \frac12n - d\rfloor} \,.
	\end{equation}
	
	\subsection{Case D: Lopsided Colorings}
	\label{sec: lopsided}
	
	We turn to the `lopsided' case where the largest color class exhausts around half of the vertices, but the remaining classes are quite small.
	To be precise, we want to detect $k$-colorings whose size profile $(s_1,\ldots, s_k)$ satisfies
	\begin{equation}
		\label{eq: lopsided}
		\tfrac12n - d < s_1 \leq \tfrac12n + d \text{ and } s_2 < 2d,
	\end{equation}
	and where $S_1$ is inclusion-maximal.
	(This is the only one of our algorithms that uses the assumption on $S_1$.)
	
	\medskip
	
	We will also require $d=\delta n$ for some constant $\delta > 0$.
	
	\subsubsection*{Algorithm}
	Given $G$, $k$, and $d$, the algorithms proceeds as follows.
	Use the algorithm of Bron and Kerbosch~\cite{BronK1973} to list every maximal independent set $X$ of size $i$ in the range $\frac12n-d\leq i\leq \frac12n + d$.
	For each such $X$ construct the graph $G'=G[V\setminus X]$ with $n'=n-i$ many vertices.
	Use Lemma~\ref{lem: 3waykcol} on $G'$ with parameter
	\begin{equation}
	\label{eq: nud}
	\nu =\frac13 + \frac{2d}{n-i}=\frac13 +\frac{2\delta n}{n-i}\,.
	\end{equation}
	to determine if it admits a $(k-1)$-coloring.
	
	\subsubsection*{Analysis}
	
	To see correctness, it is clear that if $G$ is $k$-colorable with color profile $(s_1,\ldots, s_k)$ then $G'$ is indeed $(k-1)$-colorable with the corresponding color classes $S_2,\dots,S_k$.
	But we need to argue that the balancing condition in Lemma~\ref{lem: 3waykcol}  applies.
	To this end we need to show that the color classes $S_2,\ldots, S_k$ admit a partition into three roughly equal-sized parts.
	An easy greedy procedure establishes that such a partition exists:
	Add $S_2,\ldots,S_r$ to the first part $P_1$ until $s_2+\cdots+s_r +s_{r+1}> \frac13n'+2d$.
	Similarly, add $S_{r+1},\ldots , S_t$ to the second part $P_2$ until $s_{r+1}+\cdots +s_t + s_{t+1}> \frac13n'+2d$. 
	The remaining sets go into $P_3$.
	Since $s_{r+1},s_{t+1}< 2d$ the two first parts add up to at least $\frac23n'$.  
	Thus, $P_3$ adds up to at most $\frac13n'$.
	In particular, none of the parts is larger than 
	\( \nu n'=\frac13 n' + 2d\).
	Hence we can use Lemma~\ref{lem: 3waykcol} with parameter $\nu$ to see if  $G'$ admits a $(k-1)$-coloring.
	
	\medskip
	
	To analyze the running time, it is known that the Bron--Kerbosch algorithm lists all 
	maximal independent sets within time proportional to the Moon-Moser bound of their maximum number, $3^{n/3}$~\cite{MoonM1965}, see~\cite{TomitaTT2006}.
	We discard those $X$ with $|X|< n/2-d$ or $|X|> n/2 + d$.
	To see how many sets are left, we need a bound on the number of maximal independent sets $\operatorname{\#mis}_i(G)$ of size $i$ in a graph $G$.
	\begin{lemma}
	\label{lem: mis}
		The number of maximal independent sets in an $n$-vertex graph $G$ is at most $2^{n-\alpha(G)}$, where $\alpha(G)$ is the size of the largest independent set in the graph.
	
	\end{lemma}

        \begin{proof}	
		Let $X$ be a maximum independent set in $G$, and let $Y=V(G)\setminus X$.
		Now observe that there cannot be two different maximal independent sets $A$ and $B$, possibly of different sizes, with the same projection on $Y$, i.e., with $Z=A\cap Y=B\cap Y$. Suppose there was, and let $a\in A\setminus B$ be any element and note that $a\in X$. Clearly, $\{a\}\cup Z$ is an independent set since $a\in A$. We also have that $\{a\}\cup  (B\cap X)$ is an independent set since $X$ is an independent set. This implies that $\{a\}\cup B$ is an independent set which contradicts the maximality of $B$.
		
		Hence, for every $Z\subseteq Y$, there is at most one maximal independent set $W$ with $Z=W\cap Y$. Since $Y$ has $2^{n-\alpha(G)}$ subsets, the bound follows.
	\end{proof}
	
	In particular, we have
	\begin{equation}
		\label{eq: misboundapx}
		\operatorname{\#mis}_i(G)\leq 2^{n-i}.
	\end{equation}
	Sharper bounds for maximal independent sets of size less than $n/2$ are given by Byskov~\cite{Byskov2004}, but we do not need them here.
	
	For each $X$ with $\frac12n-d\leq i\leq  \frac12n + d$, we apply the procedure behind~Lemma~\ref{lem: 3waykcol} on the $(n-i)$-vertex graph $G[V\setminus X]$, with $\nu$ given as in \eqref{eq: nud}.
	This takes time $\binom{n-i}{\lfloor \nu (n-i)\rfloor}^{1+\epsilon}$.
	In total over all salient $X$, we have
	\begin{equation}
		\sum_{i=\lceil n/2-d\rceil}^{\lfloor n/2+d\rfloor } 
		\operatorname{\#mis}_{i}(G)\binom{n-i}{\lfloor (n-i)/3+2d\rfloor }^{1+\epsilon} \leq
		2^{n/2+d}\binom{\lfloor \frac12n+d\rfloor}{\lfloor n/6+7d/3\rfloor}^{1+\epsilon}\,.
	\end{equation}
	A quick estimate shows that our initial execution of the Bron--Kerbosch algorithm will be dwarfed by these running times, so we can disregard it and estimate the total time for the algorithm in this section to within a polynomial factor or
	\[ t_{\mathrm d}(n,\delta n) = 2^{n/2 +\delta n}\binom{\lfloor \frac12n + \delta n\rfloor}{\lfloor \frac16n +\frac73 \delta n\rfloor}^{1+\epsilon}
	\,.\]
	
	\subsection{Case E: Semi-Balanced Colorings}
	\label{sec:balanced}
	
	We are left with the task of detecting $k$-colorings whose size profile $(s_1,\ldots, s_k)$ satisfies
	\begin{equation}\label{eq:balanceconstraint}
		s_1 \leq \tfrac12n-d \text{ and } s_1+s_2+s_3+s_4 \leq n-6d\,.
	\end{equation}
	We also assume $d=\delta n$ for constant $\delta > 0$.
	
	These inequalities will be enough to guarantee that the colors can be partitioned in three roughly balanced parts.
	Thus, the algorithm for this case is very simple:

	\subsubsection*{Algorithm}
	Given $G$, $k$, and $d=\delta n$, apply Lemma~\ref{lem: 3waykcol} with $\nu=\tfrac12-\delta$.

	\subsubsection*{Analysis}
	To show that the conditions for Lemma~\ref{lem: 3waykcol} apply,  we need to show that the color classes can be partitioned in three parts, none of which is greater than $\frac{n}{2}-\delta n$ in size.
	
	\begin{lemma}
		\label{lem: semi-balanced}
		Let $s_1,\ldots, s_k$ be numbers in nonascending order with
		\(
		s_1+\cdots+s_k = n
		\).
		If the numbers satisfy \eqref{eq:balanceconstraint} then 
		they can be partitioned into three sets, each of which sums to at most $\tfrac12 n - d\,.$
		
	\end{lemma}

        \begin{proof}
		Consider the following procedure.
		\begin{enumerate}
			\item Arbitrarily place $s_1$, $s_2$, and $s_3$ into $P_1$, $P_2$, $P_3$.
			Let $p_i$ denote the sum of the elements in part $P_i$.
			Invariantly, we will maintain $p_i\leq \tfrac12 n - d$.
			
			\item In ascending order, place $s_k, s_{k-1}, \ldots$ arbitrarily into the first two parts $P_1$ and $P_2$ as long as the invariant is maintained.
			
			\item Place the remaining numbers into the third part $P_3$.
		\end{enumerate}
		
		We need to argue that in step 3, the remaining numbers, say
		$s_4,\ldots, s_j$, fit into $P_3$ without violating the invariant.
		Since after step 2, the value $s_j$ does not fit in the first two parts, we have
		\[ 
		\min \{p_1,p_2\} + s_j > \tfrac12 n-d\,.
		\]
		Using an averaging argument and $p_1+p_2+p_3+r = n$, we can write
		\[
		\tfrac 12 n- d < \frac{p_1 + p_2}{2} + s_j = \frac{n - p_3 - r}{2}+ s_j\,,
		\]
		which simplifies to
		\[
		p_3 + r < 2 d+ 2s_j\,.
		\]
		From the sum constraint, we have
		\[
		s_j \leq s_4 \leq \tfrac14 n-6d\,.
		\]
		Combining these two expressions we arrive at
		\[
		p_3 + r < 2d+\tfrac12 n- 3d = \tfrac 12 n -d\,,
		\]
		so the remaining elements do indeed fit into $P_3$.
	\end{proof}

	The running time, for any $\epsilon>0$, is within a polynomial factor of
	\[
	t_{\mathrm e}(n,\delta n) = 
	\binom{n}{\lfloor\frac12n-\delta n\rfloor}^{1+\epsilon}.
	\]
	
	\subsection{Algorithm for Chromatic Number}
	\label{sec:summary}
	In this section we prove Theorem~\ref{thm: chromaticnumber}. 
	
	Given a graph $G$, we iterate over $k\in\{1,\ldots,n\}$ in ascending order.
	For each $k$, we run the applicable subset of the five algorithms from Sections~\ref{sec:chi3}--\ref{sec:balanced} with $d=\frac1{145} n$.
	If any detects a $k$-coloring, we return $k$ as the chromatic number of the graph.
	
	To establish correctness, assume there exists a $k$-coloring of $G$.
	By Lemma~\ref{lem:S1mis}, there also exists a $k$-coloring where the largest color class $S_1$ is a maximal independent set.
	The size profile of this coloring satisfies at least one of the five cases in Lemma~\ref{lem:cases}.
	Thus, one of the five algorithms in Sections~\ref{sec:chi3}--\ref{sec:balanced} detects it.
	
	\subsection{Running Times}
	\label{sec:chromatictimes}
	We set $d=\delta n$  and assume that $n/2\delta$ is an integer to avoid tedious computations.
	Using the entropy bound~\eqref{eq:binomentropy},
	we estimate the running times for each case.
	First, \[
	t_{\mathrm b}(n,\delta n)  = 
	\sum_{i=1}^{6\delta n}\binom{n}{i} t_4(n-i)
	\leq
	\sum_{i=1}^{6\delta n} 2^{H(i/n)n} t_4(n-i)
	\,.\]
	In the range $0\leq i\leq\frac14n$ the largest term occurs for the largest $i$.
	Thus we have
	\begin{equation} \label{eq:tb} 
	  t_{\mathrm b}(n, \delta n) \leq	
	n\binom{n}{6 \delta n} t_4 ((1-6\delta)n)
	\leq n2^{H(6\delta)n} t_4((1-6\delta)n)\,.
	\end{equation}
	We also have 
	\begin{equation}\label{eq:td}
	t_{\mathrm d}(n,\delta n) =  
	2^{(\frac12 + \delta)n}
	\binom{(\frac12 + \delta)n}{(\frac16+\frac{7}{3}\delta) n}^{1+\epsilon}
	\leq 
	\bigl(2\cdot 2^{H[
		(\frac16 + \frac73\delta)/(\frac12+\delta)]}\bigr)^{(\frac12 + \delta)\cdot(1+\epsilon)\cdot n},
	\end{equation}
	and 
	\begin{equation}\label{eq:te}
	  t_{\mathrm e}(n,\delta n)  = \binom{n}{(\frac12 -\delta)n}^{1+\epsilon}\leq 2^{H(\frac12 -\delta)\cdot (1+\epsilon)\cdot  n}\,.
	\end{equation}
	
	For deterministic $3$-coloring and $4$-coloring we use the bounds $t_3(r)<1.3289^r$~\cite{BeigelE2005} and $t_4(r)<1.7215^r$~\cite{Clinch2024}.
	With $\delta= \frac1{145}$ we arrive at
	\begin{gather*}
		t_{\mathrm a}(n) \leq 1.3289^n,\quad
		t_{\mathrm b}(n,\tfrac1{145}n) \leq 1.9998^n,\quad
		t_{\mathrm d}(n,\tfrac1{145}n) \leq 1.98^{(1+\epsilon)\cdot n},\\
		t_{\mathrm c}(n,\tfrac1{145}n) < t_{\mathrm e}(n,\tfrac1{145}n) \leq 1.99981^{(1+\epsilon)\cdot n}\,.
	\end{gather*}
	In particular, with $\epsilon$ sufficiently small, the total running time is asymptotically
	dominated by $1.99982^n$.

\section{Deterministic Set Cover under the Asymptotic Rank Conjecture}
	\label{sec: set cover under arc}
	In this section we first present the proof of Theorem~\ref{thm:fine-3part}. We next use it via  Lemma~\ref{lem:balancedkcover} to prove Theorem~\ref{thm: setcover}.
		
	\subsection{Three-Way Partitioning}
	
	We first consider a ``block version'' of the partitioning problem and then reduce the general problem to it. In the ideal block version, the input universe $[n]$ would be partitioned into blocks $B_1 ,\ldots, B_r$ of equal size $b=n/r$, and all sets $A$ in the input set families would contain exactly $|A| / r$ elements per block.
	To enable our derandomization approach, small deviations from this ideal need to be allowed; see Figure~\ref{fig: block-balanced} for an illustration:
	A set $A\subseteq[n]$ is called $(\delta,r)$-balanced onto another set $B$ if $|A\cap B| \in (1\pm \delta)\cdot|A|/r$. In other words, intersecting $A$ with $B$ roughly reduces its size by a factor $1/r$.
	
	Formally, the \emph{$(\delta,r)$-block-balanced partitioning problem} obtains the following input:
	\begin{itemize}
		\item parameters $\delta > 0$ and $r \in \mathbb N$,
		\item a \emph{block partition} of $[n]$ with $B_{1} \dcup\ldots\dcup B_{s} = [n]$ for $s \in \mathbb N$ such that $[n]$ is $(\delta,r)$-balanced onto each $B_j$,
		\item three $\nu$-bounded set families $\mathcal{F}_1,\mathcal{F}_2,\mathcal{F}_3$ over a universe $[n]$, such that each set $A \in \mathcal{F}_1 \cup \mathcal{F}_2 \cup \mathcal{F}_3$ is $(\delta,r)$-balanced onto each $B_j$.
		 
	\end{itemize}
	The computational task is to decide the existence of $A_1 \dcup A_2 \dcup A_3 = [n]$ with $A_i \in \mathcal F_i$ for $i\in[3]$.
	We show near-optimal algorithms for this problem under the asymptotic rank conjecture.
	\begin{figure}
                \begin{center}
		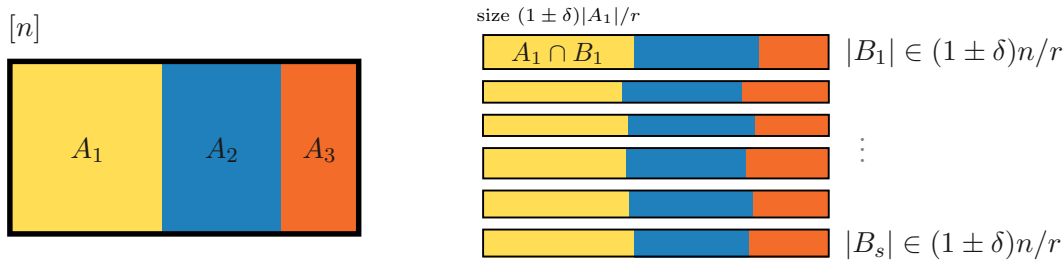
                \end{center}
		\caption{On the left, a three-way partitioning $A_1 \dcup A_2 \dcup A_3 = [n]$ is shown horizontally. On the right, the universe $[n]$ is vertically partitioned into blocks $B_1, \ldots , B_s$. Each set $A_i$ for $i\in[3]$ is $(\delta,r)$-balanced onto each block.}
		\label{fig: block-balanced}
		\centering	
	\end{figure}
	\begin{lemma}[Solving block-balanced partitioning]
	\label{lem: block-balanced partitioning}
		Let $1/3 \leq \nu < 1/2$ and let $0< \delta < \frac{1}{2\nu}-1$. 
		Assuming the asymptotic rank conjecture,
		the $(\delta,r)$-block-balanced partitioning problem for $\nu$-bounded set families and $n/r \geq 3 / \delta$ can be solved in time $O({n \choose \lfloor\nu n\rfloor}^\kappa)$ with $\kappa \to 1$ as $\delta \to 0$.
	\end{lemma}

	\begin{proof}
	Let $\mathcal F_1, \mathcal F_2 , \mathcal F_3$ be set families over $[n]$, let $B_{1} \dcup\ldots\dcup B_{s} = [n]$ be a block partition, and let $r \in \mathbb N$. Write $b := n/r$ for the ``ideal'' block size.
	With $\nu' := (1+\delta)\nu < 1/2$, 
	our input assumptions imply that $|A \cap B_j| \leq \nu' b$ for all $j\in [s]$ and sets $A$ in the input families.
	Let $B := \max_j |B_j| \leq (1+\delta)b$ 
	and consider the tensor given by 
	\[
	T=\sum_{k=1}^{B}\sum_{\substack{I\dcup J\dcup K=[k]\\
	|I|,|J|,|K|\leq\lceil\nu' b\rceil
	}
	}X_{k,I}Y_{k,J}Z_{k,K},
	\]
	which is the direct sum of partitioning tensors up to universe size $B$. 
	Using $\nu' < 1/2$ in the first inequality, handling rounding in the second inequality, and using $B \geq b \geq 3/\delta$, the number of dimensions per leg is
	\[
	d := \sum_{k=0}^{B}\sum_{\ell=0}^{\lceil\nu' b\rceil}{k \choose \ell}
	\ \leq\ B^{2}{B \choose \lceil\nu' b\rceil}
	\ \leq\ B^{3}{B \choose \lfloor\nu' b\rfloor}
	\ \leq\ B^{3}2^{H(\nu')B}
	\ \leq\ 2^{(H(\nu') + \delta)B}
	\ \leq\ 2^{\beta' H(\nu) B},
	\]
	with $\beta' \to 1$ as $\delta \to 0$.
	It follows that $d^s \leq 2^{\beta H(\nu) n}$ with $\beta \to 1$ as $\delta \to 0$.
	To detect a partitioning, evaluate $T^{\otimes s}$
	on binary encodings $\xi_i$ of the set families $\mathcal F_i$:
	For each $A\in\mathcal F_1$, set 
	\[
		X_{\kappa}\gets 1 \mbox{ for } \kappa=(|B_{1}|,A\cap B_{1},\,\ldots,\,|B_{r}|,A\cap B_{s}),
	\] 
	and set all other variables $X_{\kappa'} \gets 0$. 
	This yields $\xi_1$ and similarly $\xi_2$ and $\xi_3$. 
	We have $T^{\otimes s}(\xi_1,\xi_2,\xi_3) \neq 0$
	iff there exist sets $A_1 \dcup A_2 \dcup A_3 = [n]$ with $A_i \in \mathcal F_i$ for $i\in[3]$ and 
	$|A_i \cap B_j| \leq \lceil \nu' b \rceil$ for $j\in [s]$, where the last condition guaranteed by the input.
	By Lemma~\ref{lem:trilinear-arc} with $\epsilon = \beta-1$, we can evaluate $T^{\otimes s}$ in time $O(d^{(1+\epsilon)s}) = 2^{\beta^2 H(\nu) n }$, with $\beta^2 \to 1$ as $\delta \to 0$.
	\end{proof}
	Then we reduce the general partitioning problem to its block-balanced version.
	Our reduction is based on \emph{balancing families}, which we define to be moderately small families $\mathcal B$ of partitions $B_0 \cup \ldots \cup B_s = [n]$ with a particular balancing property: 
	For every possible $\nu$-bounded input to the three-partitioning problem, there exists at least one block partition in $\mathcal B$ such that the sets in the input are $(\delta,r)$-balanced onto every block, with the exception of one small ``fault block'' $B_0$.
	\begin{Definition}\label{def: block-balancing}
		For $\delta, r, \nu > 0$, a \emph{$(\delta,r,\nu)$-balancing family} is a set $\mathcal B$ of block partitions of $[n]$ such that each $(B_0,B_1,\ldots,B_s) \in \mathcal B$ satisfies 
		\begin{itemize}
			\item $|B_0| \leq \delta n$ (i.e., the fault block $B_0$ is small) and
			\item $|B_i|\in (1\pm \delta)n/r$ for $i=1,\ldots , s$ (i.e., the universe $[n]$ is $(\delta,r)$-balanced onto each $B_i$), and
		\end{itemize}
		for every partition $A_1 \dcup A_2 \dcup A_3 = [n]$ with $\nu$-bounded sets $A_i$, there exists at least one block partition $(B_0,B_1,\ldots,B_s) \in \mathcal B$ such that each $A_i$ is $(\delta,r)$-balanced onto $B_1, \ldots , B_s$.
	\end{Definition}
	We prove that sufficiently small balancing families exist and can be constructed efficiently. To this end, 
	we use pairwise independent hash functions.
	For $r\in \mathbb N$, a family $\mathcal H$ of functions $h:[n] \to [r]$ is \emph{pairwise independent} if, for any fixed $i,i'\in[n]$ and $j,j'\in[r]$ with $i \neq i'$, a randomly drawn $h \in\mathcal{H}$ satisfies $h(i)=j$ and $h(i')=j'$ with probability $1/r^{2}$.
	For $r \leq n$ with $r$ prime, there is a pairwise independent family $\mathcal H = \{h_1, \ldots, h_{|\mathcal H|} \}$ of hash functions $h:[n] \to [r]$ with $|\mathcal H| \leq O(n^2r^2)$
	such that $h_i(x)$ can be computed in polynomial time from $i$ and $x \in [n]$, see standard textbooks~\cite{MitzenmacherUpfal05}.

	\begin{lemma}[Constructing balancing families]
		\label{lem: construct-balancer}
			For every $\delta > 0$ and $b' \geq 3 \delta^{-3} (1 - 2\nu)^{-2}$, 
			we can compute $b' \leq b \leq 2b'$ and
			a $(\delta,n/b,\nu)$-balancing family $\mathcal B$ in time $\bO(2^{\delta' n})$ with $\delta' \to 0$ as $\delta \to 0$.
		\end{lemma}
		
	\begin{proof}
		Let $\kappa = 1 - 2\nu$.
		By Bertrand's postulate, there is some $b'\leq b \leq 2b'$
		such that $r := n/b$ is prime.
		Let $\mathcal H$ be an explicit pairwise independent family of hash functions from $[n]$ to $[r]$.
	The proof hinges upon the following central claim:
	Given any sets $A_1,A_2,A_3 \subseteq[n]$ of size $\geq \kappa n$ each, there is some $h \in \mathcal H$ and set $Q\subseteq [r]$ of size $\leq \delta r$ such that each $A_i$ is $(\delta,r)$-balanced onto $h^{-1}(j)$ for all $j\in [r]\setminus Q$.
	
	To prove this claim, fix a set $A\subseteq [n]$ with $|A| \geq \kappa n$ and fix $j\in[r]$. The random variable $X=|h^{-1}(j)\cap A|$ has mean $\mu = |A| / r$ and variance $\sigma^2 = |A| \frac{r-1}{r^2} \leq \frac{b(r-1)}{r} \leq b$.
		The set $A$ is $(\delta,r)$-balanced onto $h^{-1}(j)$ iff $|X-\mu|\leq t$ with $t = \delta \mu \geq \delta \kappa n / r = \delta \kappa b$.
		By Chebyshev's inequality, the event $|X-\mu| > t$ occurs with probability at most
		\[
			p \leq \frac{\sigma^2}{t^2} 
			\leq \frac{b}{t^{2}}
			\leq \frac{1}{\delta^2 \kappa^2 b}
			\leq \delta / 3,
		\]
		where we used the choice of $b \geq b'$ from the beginning.
		By a union bound, the probability of any of $A_1,A_2,A_3$ not being $\delta$-balanced onto $h^{-1}(j)$ is at most $\delta$.
		In expectation, there are thus at most $\delta r$ indices $j$ such that any of $A_1,A_2,A_3$ is not $\delta$-balanced onto $h^{-1}(j)$.
		By averaging, this bound is attained by some $h \in \mathcal H$.
		This concludes the proof of the central claim.

	To construct $\mathcal B$, 
	for each $h\in \mathcal H$ and $Q\subseteq [r]$ of size $\leq \delta r$, add a partition $B_{(h,Q)}$ into $\mathcal B$ with zero-th block $\bigcup_{j\in Q} h^{-1}(j)$, followed by blocks $h^{-1}(j_i)$ for $j \in [r]\setminus Q$.
	This can be completed in $\btO(2^{r}) = \btO(2^{cn})$ time with $c = 1/b \leq 1/b' \leq \delta^3 \kappa^2 / 3$, so $c \to 0$ as $\delta \to 0$ and $\kappa$ is fixed.

	We check that $\mathcal B$ is indeed balancing:
	Let $A_1\dcup A_2 \dcup A_3 = [n]$ be a partitioning with $|A_i| \leq \nu n$ for each $i$. This implies $|A_i|\geq \kappa n$ for all $i$.
	By the above claim, there is some $h\in \mathcal H$ and $Q \subseteq [r]$ such that each $A_i$ is $(\delta,r)$-balanced onto $h^{-1}(j)$ for $j\in [r]\setminus Q$, so $A_i$ is $(\delta,r)$-balanced onto all blocks of $B_{(h,Q)} \in \mathcal B$ except for the zero-th block, as required.
	\end{proof}

	Next, we show how balancing families enable a deterministic reduction to the block-balanced version, thus finally proving Theorem~\ref{thm:fine-3part}.
	\begin{proof}[\emph{Proof of Theorem~\ref{thm:fine-3part}}]
		Choose $b' := 3 \delta^{-3} (1 - 2\nu)^{-2}$.
		Let $b$ be the number and $\mathcal B$ be the $(\delta,n/b,\nu)$-balancing family obtained from Lemma~\ref{lem: construct-balancer} with these parameters,
		and let $r = n/b$.

		We show how to solve the partitioning problem for $\nu$-bounded sets with $|\mathcal B| 3^{\delta n}$ calls to its block-balanced version. The overall running time is $O(|\mathcal B| 3^{\delta n} \cdot {n \choose \lfloor\nu n\rfloor}^{1+\epsilon'}) = O({n \choose \lfloor\nu n\rfloor}^{1+\epsilon})$ with $\epsilon', \epsilon \to 0$ as $\delta \to 0$.
		
		The algorithm proceeds as follows:
		First iterate over all $(B_0,B_1,\ldots,B_s) \in \mathcal B$
		and all $3$-partitionings $(X_1,X_2,X_3)$ of $B_0$.
		In each iteration, keep only those sets $A$ in $\mathcal F_i$ with $A\cap B_0 = X_i$, for $i\in[3]$.
		Since $|B_0| \leq \delta n$, this incurs at most $|\mathcal B| \cdot 3^{\delta n}$ iterations.
		In each of them, use Lemma~\ref{lem: block-balanced partitioning} to solve the $(\delta,r$)-block-balanced partitioning problem with blocks $B_1,\ldots,B_s$, universe $[n] \setminus B_0$, and keeping only sets in the input families that are $(\delta,r)$-balanced onto the blocks $B_1,\ldots ,B_s$. The condition $n/r \geq 3/\delta$ required in Lemma~\ref{lem: block-balanced partitioning} is satisfied.
		Then answer positively if any of these tests succeeds.

		This procedure cannot create false positives.
		Given a valid partitioning $A_1 \dcup A_2 \dcup A_3$ of $[n]$ with sets of size $\leq \nu n$,
		the definition of $(\delta,r,\nu)$-balancing families implies that for at least one block partition in $(B_0,B_1,\ldots,B_s) \in \mathcal B$,
		every $A_i$ is $(\delta,r)$-balanced onto the blocks $B_1,\ldots,B_s$.
		The overall running time of this algorithm is $O(|\mathcal B| \cdot 3^{\delta n} \cdot {n \choose \lfloor\nu n\rfloor}^\kappa)$ with $\kappa \to 1$ as $\delta \to 0$, which is $O(|\mathcal B| \cdot {n \choose \lfloor\nu n\rfloor}^{1+\epsilon})$ with $\epsilon \to 0$ as $\delta \to 0$.
	\end{proof}

\subsection{Set Cover}
We next turn to prove Theorem~\ref{thm: setcover}. To get an intuition for the $1/4$ bound in the theorem, note that an instructive situation where the above algorithm fails is to find a $4$-cover of the universe $[4]$ from the family $\mathcal F =\bigl\{\{1\}, \{2\}, \{3\}, \{4\}\bigr\}$.

	\begin{proof}[Proof of Theorem~\ref{thm: setcover}]
	It suffices by Lemma~\ref{lem:balancedkcover} to show that any $t$-cover from a set family $\mathcal F$ of subsets each of size bounded by $\delta n$ for $\delta< 1/4$, is $(1/2-\kappa)$-balanced for some $\kappa>0$. This can be rephrased as a partition problem of a sequence of reals: Given $0<a_1,\ldots,a_t$ with $\sum_{i=1}^t a_i=1$, and $\forall i:a_i\leq \delta$, show that they can be partitioned into three sets $A_1,A_2$, and $A_3$, such that, for $i=1,2,3$, 
	\begin{equation}
	\label{eq: size-bound}
	\sum_{a\in A_i} a\leq \tfrac12-\kappa.
	\end{equation}
	
	We give a constructive proof. Let $w(A_i)=\sum_{a\in A_i} a$ and set $\kappa=(\tfrac12-2\delta)/3$. Observe that $\kappa>0$. Insert $a_1,\ldots, a_t$ in order greedily into the three sets $A_1,A_2,A_3$. That is, start by filling up $A_1$ until it is no longer possible to add the next element without violating the size bound~\eqref{eq: size-bound}, and continue filling up $A_2$ likewise, and finally $A_3$ with the remaining elements. Note that $\tfrac12-\kappa-w(A_1)\leq \delta$ and $\tfrac12-\kappa-w(A_2)\leq \delta$ because of the greedy procedure. Hence, $w(A_1)+w(A_2)\geq 1-2\delta-2\kappa$. This means that $w(A_3)=1-w(A_1)-w(A_2)\leq 2\delta+2\kappa=\tfrac12-\kappa$. Consequently, all three sets $A_1,A_2,A_3$ obey the size bound~\eqref{eq: size-bound}.
	\end{proof}

        \section{Conclusion}

	We have shown that, under the asymptotic rank conjecture, the chromatic number algorithm of Bj\"orklund, Husfeldt, and Koivisto~\cite{BjorklundHK2009} can be improved.
	However, the algorithm of \cite{BjorklundHK2009} solves the more general \emph{counting} problem of determining the number of $k$-colorings, and it computes the coefficients of the \emph{chromatic polynomial} of the input graph.
	We do not know how to count colorings or compute the chromatic polynomial faster under the asymptotic rank conjecture.
	\subsection*{Acknowledgements}
	We thank Nutan Limaye for inviting KP to ITU Copenhagen, where the authors first all met. We also thank Cornelius Brand for discussions in an early stage of the project.
	\bibliographystyle{abbrv}
	\bibliography{paper}
\end{document}

%% file: block-balancing.pdf_tex
\begingroup%
  \makeatletter%
  \providecommand\color[2][]{%
    \errmessage{(Inkscape) Color is used for the text in Inkscape, but the package 'color.sty' is not loaded}%
    \renewcommand\color[2][]{}%
  }%
  \providecommand\transparent[1]{%
    \errmessage{(Inkscape) Transparency is used (non-zero) for the text in Inkscape, but the package 'transparent.sty' is not loaded}%
    \renewcommand\transparent[1]{}%
  }%
  \providecommand\rotatebox[2]{#2}%
  \newcommand*\fsize{\dimexpr\f@size pt\relax}%
  \newcommand*\lineheight[1]{\fontsize{\fsize}{#1\fsize}\selectfont}%
  \ifx\svgwidth\undefined%
    \setlength{\unitlength}{392.59375925bp}%
    \ifx\svgscale\undefined%
      \relax%
    \else%
      \setlength{\unitlength}{\unitlength * \real{\svgscale}}%
    \fi%
  \else%
    \setlength{\unitlength}{\svgwidth}%
  \fi%
  \global\let\svgwidth\undefined%
  \global\let\svgscale\undefined%
  \makeatother%
  \begin{picture}(1,0.2888121)%
    \lineheight{1}%
    \setlength\tabcolsep{0pt}%
    \put(0,0){\includegraphics[width=\unitlength,page=1]{block-balancing.pdf}}%
    \put(0.0608489,0.11917933){\color[rgb]{0.10196078,0.10196078,0.10196078}\makebox(0,0)[lt]{\lineheight{1.25}\smash{\begin{tabular}[t]{l}$A_1$\end{tabular}}}}%
    \put(0.18964717,0.11917931){\color[rgb]{0.10196078,0.10196078,0.10196078}\makebox(0,0)[lt]{\lineheight{1.25}\smash{\begin{tabular}[t]{l}$A_2$\end{tabular}}}}%
    \put(0.28358718,0.11917933){\color[rgb]{0.10196078,0.10196078,0.10196078}\makebox(0,0)[lt]{\lineheight{1.25}\smash{\begin{tabular}[t]{l}$A_3$\end{tabular}}}}%
    \put(0.00078138,0.23634501){\color[rgb]{0.10196078,0.10196078,0.10196078}\makebox(0,0)[lt]{\lineheight{1.25}\smash{\begin{tabular}[t]{l}$[n]$\end{tabular}}}}%
    \put(0.80452482,0.21286286){\color[rgb]{0.10196078,0.10196078,0.10196078}\makebox(0,0)[lt]{\lineheight{1.25}\smash{\begin{tabular}[t]{l}$|B_1|\in(1\pm \delta)n/r$\end{tabular}}}}%
    \put(0.45091823,0.25069814){\color[rgb]{0,0,0}\makebox(0,0)[lt]{\lineheight{1.25}\smash{\begin{tabular}[t]{l}\tiny{size $(1\pm \delta)|A_1|/r$}\end{tabular}}}}%
    \put(0,0){\includegraphics[width=\unitlength,page=2]{block-balancing.pdf}}%
    \put(0.48389557,0.2124709){\color[rgb]{0.10196078,0.10196078,0.10196078}\makebox(0,0)[lt]{\lineheight{1.25}\smash{\begin{tabular}[t]{l}\small{$A_1\cap B_1$}\end{tabular}}}}%
    \put(0.80383773,0.02826795){\color[rgb]{0.10196078,0.10196078,0.10196078}\makebox(0,0)[lt]{\lineheight{1.25}\smash{\begin{tabular}[t]{l}$|B_s|\in(1\pm \delta)n/r$\end{tabular}}}}%
    \put(0.8170807,0.11495696){\color[rgb]{0.50196078,0.50196078,0.50196078}\makebox(0,0)[lt]{\lineheight{1.25}\smash{\begin{tabular}[t]{l}$\vdots$\end{tabular}}}}%
  \end{picture}%
\endgroup%

%% file: paper.bbl
\begin{thebibliography}{10}

\bibitem{BeigelE2005}
R.~Beigel and D.~Eppstein.
\newblock 3-coloring in time ${O}(1.3289^n)$.
\newblock {\em Journal of Algorithms}, 54(2):168--204, 2005.

\bibitem{BjorklundH2006}
A.~Bj{\"{o}}rklund and T.~Husfeldt.
\newblock Exact algorithms for exact satisfiability and number of perfect
  matchings.
\newblock In M.~Bugliesi, B.~Preneel, V.~Sassone, and I.~Wegener, editors, {\em
  Automata, Languages and Programming, 33rd International Colloquium, {ICALP}
  2006, Venice, Italy, July 10-14, 2006, Proceedings, Part {I}}, volume 4051 of
  {\em Lecture Notes in Computer Science}, pages 548--559. Springer, 2006.

\bibitem{BjorklundHKK10}
A.~Bj{\"{o}}rklund, T.~Husfeldt, P.~Kaski, and M.~Koivisto.
\newblock Trimmed moebius inversion and graphs of bounded degree.
\newblock {\em Theory Comput. Syst.}, 47(3):637--654, 2010.

\bibitem{BjorklundHK2009}
A.~Bj{\"{o}}rklund, T.~Husfeldt, and M.~Koivisto.
\newblock Set partitioning via inclusion-exclusion.
\newblock {\em {SIAM} J. Comput.}, 39(2):546--563, 2009.

\bibitem{BjorklundK24}
A.~Bj{\"{o}}rklund and P.~Kaski.
\newblock The asymptotic rank conjecture and the set cover conjecture are not
  both true.
\newblock In B.~Mohar, I.~Shinkar, and R.~O'Donnell, editors, {\em Proceedings
  of the 56th Annual {ACM} Symposium on Theory of Computing, {STOC} 2024,
  Vancouver, BC, Canada, June 24-28, 2024}, pages 859--870. {ACM}, 2024.

\bibitem{BronK1973}
C.~Bron and J.~Kerbosch.
\newblock Finding all cliques of an undirected graph.
\newblock {\em Commun. ACM}, 16(9):575--577, sep 1973.

\bibitem{BurgisserCS2013}
P.~B\"{u}rgisser, M.~Clausen, and M.~A. Shokrollahi.
\newblock {\em Algebraic Complexity Theory}.
\newblock Springer Science \& Business Media, 2013.

\bibitem{Byskov2003}
J.~M. Byskov.
\newblock Algorithms for k-colouring and finding maximal independent sets.
\newblock In {\em Proceedings of the Fourteenth Annual {ACM-SIAM} Symposium on
  Discrete Algorithms, January 12-14, 2003, Baltimore, Maryland, {USA}}, pages
  456--457. {ACM/SIAM}, 2003.

\bibitem{Byskov2004}
J.~M. Byskov.
\newblock Enumerating maximal independent sets with applications to graph
  colouring.
\newblock {\em Operations Research Letters}, 32(6):547--556, 2004.

\bibitem{ByskovMS2005}
J.~M. Byskov, B.~A. Madsen, and B.~Skjernaa.
\newblock On the number of maximal bipartite subgraphs of a graph.
\newblock {\em Journal of Graph Theory}, 48(2):127--132, 2005.

\bibitem{ChristandlVZ2021}
M.~Christandl, P.~Vrana, and J.~Zuiddam.
\newblock Barriers for fast matrix multiplication from irreversibility.
\newblock {\em Theory Comput.}, 17:Paper No. 2, 32, 2021.

\bibitem{Christofides1971}
N.~Christofides.
\newblock An algorithm for the chromatic number of a graph.
\newblock {\em Comput. J.}, 14(1):38--39, 1971.

\bibitem{Clinch2024}
K.~Clinch, S.~Gaspers, A.~Saffidine, and T.~Zhang.
\newblock A piecewise approach for the analysis of exact algorithms, 2024.
\newblock arXiv:2402.10015.

\bibitem{ConnerGLV2022}
A.~Conner, F.~Gesmundo, J.~M. Landsberg, and E.~Ventura.
\newblock Rank and border rank of {K}ronecker powers of tensors and
  {S}trassen's laser method.
\newblock {\em Comput. Complexity}, 31(1):Paper No. 1, 40, 2022.

\bibitem{ConnerGLVW2021}
A.~Conner, F.~Gesmundo, J.~M. Landsberg, E.~Ventura, and Y.~Wang.
\newblock Towards a geometric approach to {S}trassen's asymptotic rank
  conjecture.
\newblock {\em Collect. Math.}, 72(1):63--86, 2021.

\bibitem{Eppstein2003}
D.~Eppstein.
\newblock Small maximal independent sets and faster exact graph coloring.
\newblock {\em J. Graph Algorithms Appl.}, 7(2):131--140, 2003.

\bibitem{Fomin2007}
F.~Fomin, S.~Gaspers, and S.~Saurabh.
\newblock Improved exact algorithms for counting 3- and 4-colorings.
\newblock In {\em International Computing and Combinatorics Conference}, 2007.

\bibitem{Gartenberg1985}
P.~A. Gartenberg.
\newblock {\em Fast Rectangular Matrix Multiplication}.
\newblock PhD thesis, University of California, Los Angeles, 1985.

\bibitem{Jukna2011}
S.~Jukna.
\newblock {\em Extremal Combinatorics}.
\newblock Texts in Theoretical Computer Science. An EATCS Series. Springer,
  Heidelberg, second edition, 2011.
\newblock With applications in computer science.

\bibitem{karppa2019probabilistic}
M.~Karppa and P.~Kaski.
\newblock Probabilistic tensors and opportunistic boolean matrix
  multiplication.
\newblock In {\em Proceedings of the Thirtieth Annual ACM-SIAM Symposium on
  Discrete Algorithms}, pages 496--515. SIAM, 2019.

\bibitem{Landsberg2012}
J.~M. Landsberg.
\newblock {\em Tensors: Geometry and Applications}, volume 128 of {\em Graduate
  Studies in Mathematics}.
\newblock American Mathematical Society, Providence, RI, 2012.

\bibitem{Lawler1976}
E.~L. Lawler.
\newblock A note on the complexity of the chromatic number problem.
\newblock {\em Inf. Process. Lett.}, 5(3):66--67, 1976.

\bibitem{Meijer2023}
L.~Meijer.
\newblock 3-coloring in time ${O}(1.3217^n)$, 2023.
\newblock arXiv:2302.13644.

\bibitem{MitzenmacherUpfal05}
M.~Mitzenmacher and E.~Upfal.
\newblock {\em Probability and Computing: Randomized Algorithms and
  Probabilistic Analysis}.
\newblock Cambridge University Press, 2005.

\bibitem{MoonM1965}
J.~Moon and L.~Moser.
\newblock On cliques in graphs.
\newblock {\em Israel J. Math.}, 3:23--28, 1965.

\bibitem{Pratt2024}
K.~Pratt.
\newblock A stronger connection between the asymptotic rank conjecture and the
  set cover conjecture.
\newblock In B.~Mohar, I.~Shinkar, and R.~O'Donnell, editors, {\em Proceedings
  of the 56th Annual {ACM} Symposium on Theory of Computing, {STOC} 2024,
  Vancouver, BC, Canada, June 24-28, 2024}, pages 871--874. {ACM}, 2024.

\bibitem{Schiermeyer93}
I.~Schiermeyer.
\newblock Deciding 3-colourability in less than ${O}(1.415^n)$ steps.
\newblock In J.~van Leeuwen, editor, {\em Graph-Theoretic Concepts in Computer
  Science, 19th International Workshop, {WG} '93, Utrecht, The Netherlands,
  June 16-18, 1993, Proceedings}, volume 790 of {\em Lecture Notes in Computer
  Science}, pages 177--188. Springer, 1993.

\bibitem{Strassen1986}
V.~Strassen.
\newblock The asymptotic spectrum of tensors and the exponent of matrix
  multiplication.
\newblock In {\em 27th Annual Symposium on Foundations of Computer Science,
  Toronto, Canada, 27-29 October 1986}, pages 49--54. {IEEE} Computer Society,
  1986.

\bibitem{Strassen1988}
V.~Strassen.
\newblock The asymptotic spectrum of tensors.
\newblock {\em J. Reine Angew. Math.}, 384:102--152, 1988.

\bibitem{Strassen1991}
V.~Strassen.
\newblock Degeneration and complexity of bilinear maps: some asymptotic
  spectra.
\newblock {\em J. Reine Angew. Math.}, 413:127--180, 1991.

\bibitem{Strassen1994}
V.~Strassen.
\newblock Algebra and complexity.
\newblock In {\em First {E}uropean {C}ongress of {M}athematics, {V}ol. {II}
  ({P}aris, 1992)}, volume 120 of {\em Progr. Math.}, pages 429--446.
  Birkh\"{a}user, Basel, 1994.

\bibitem{TomitaTT2006}
E.~Tomita, A.~Tanaka, and H.~Takahashi.
\newblock The worst-case time complexity for generating all maximal cliques and
  computational experiments.
\newblock {\em Theoretical Computer Science}, 363(1):28--42, 2006.
\newblock Computing and Combinatorics.

\bibitem{WigdersonZ2023}
A.~Wigderson and J.~Zuiddam.
\newblock Asymptotic spectra: {T}heory, applications and extensions.
\newblock Manuscript dated October 24, 2023; available at
  \url{https://staff.fnwi.uva.nl/j.zuiddam/papers/convexity.pdf}, 2023.

\bibitem{WuGJSX2024}
P.~Wu, H.~Gu, H.~Jiang, Z.~Shao, and J.~Xu.
\newblock {A Faster Algorithm for the 4-Coloring Problem}.
\newblock In T.~Chan, J.~Fischer, J.~Iacono, and G.~Herman, editors, {\em 32nd
  Annual European Symposium on Algorithms (ESA 2024)}, volume 308 of {\em
  Leibniz International Proceedings in Informatics (LIPIcs)}, pages
  103:1--103:18, Dagstuhl, Germany, 2024. Schloss Dagstuhl -- Leibniz-Zentrum
  f{\"u}r Informatik.

\bibitem{Yates1937}
F.~Yates.
\newblock {\em The Design and Analysis of Factorial Experiments}.
\newblock Imperial Bureau of Soil Science, 1937.

\bibitem{Zamir21}
O.~Zamir.
\newblock Breaking the $2^n$ barrier for 5-coloring and 6-coloring.
\newblock In N.~Bansal, E.~Merelli, and J.~Worrell, editors, {\em 48th
  International Colloquium on Automata, Languages, and Programming, {ICALP}
  2021, July 12-16, 2021, Glasgow, Scotland (Virtual Conference)}, volume 198
  of {\em LIPIcs}, pages 113:1--113:20. Schloss Dagstuhl - Leibniz-Zentrum
  f{\"{u}}r Informatik, 2021.

\bibitem{Zamir23}
O.~Zamir.
\newblock Algorithmic applications of hypergraph and partition containers.
\newblock In B.~Saha and R.~A. Servedio, editors, {\em Proceedings of the 55th
  Annual {ACM} Symposium on Theory of Computing, {STOC} 2023, Orlando, FL, USA,
  June 20-23, 2023}, pages 985--998. {ACM}, 2023.

\end{thebibliography}
